\documentclass[runningheads,a4paper]{llncs}

\usepackage{amssymb}
\usepackage{graphicx}
\usepackage{epstopdf}
\usepackage{algpseudocode}
\usepackage{algorithm}
\usepackage{placeins}

\usepackage{amsmath,amsfonts,amsthm,amssymb}

\usepackage{xcolor,colortbl}

\definecolor{Gray}{gray}{0.85}
\definecolor{LightCyan}{rgb}{0.88,1,1}

\newcolumntype{a}{>{\columncolor{Gray}}c}
\newcolumntype{b}{>{\columncolor{white}}c}
\usepackage{url}
\urldef{\mailsa}\path|asad.khan@seecs.edu.pk|
\newcommand{\keywords}[1]{\par\addvspace\baselineskip
\noindent\keywordname\enspace\ignorespaces#1 }

\begin{document}

\mainmatter  

\title{Transform Domain Analysis of Sequences}

\titlerunning{Transform Domain Analysis of  Sequences}

%
%
\author{Muhammad Asad Khan%
\and Fauzan Mirza\and Amir Ali Khan}
\authorrunning{M Asad}
\authorrunning{M Asad, Fauzan Mirza, Amir A Khan}

\institute{National University of Sciences and Technology,\\
Islamabad, Pakistan\\
\mailsa\\}
%
%

\toctitle{Lecture Notes in Computer Science}
\tocauthor{Authors' Instructions}
\maketitle

\begin{abstract}
In cryptanalysis, security of ciphers vis-a-vis attacks is gauged against three criteria of complexities, i.e., computations, memory and time. Some features may not be so apparent in a particular domain, and their analysis in a transformed domain often reveals interesting patterns. Moreover, the complexity criteria in different domains are different and performance improvements are often achieved by transforming the problem in an alternate domain. 
Owing to the results of coding theory and signal processing, Discrete Fourier Transform (DFT) based attacks have proven to be efficient than algebraic attacks in terms of their computational complexity. Motivated by DFT based attacks, we present a transformed domain analysis of Linear Feedback Shift Register(LFSR) based sequence generators. The time and frequency domain behavior of non-linear filter and combiner generators is discussed along with some novel observations based on the Chinese Remainder Theorem (CRT). CRT is exploited to establish patterns in LFSR sequences and underlying cyclic structures of finite fields. Application of DFT spectra attacks on combiner generators is also demonstrated. Our proposed method saves on the last stage computations of selective DFT attacks for combiner generators. The proposed approach is demonstrated on some examples of combiner generators and is scalable to general configuration of combiner generators.
\end{abstract}

\section{Introduction}
LFSRs have been widely used for sequence generation due to their inbuilt recursive structure, faster implementations and well studied behaviour in diverse applications of communications, coding theory and cryptology. In cryptographic algorithms, the linear recurrence of LFSRs is modified by nonlinear filtering to
achieve higher linear complexities and good statistical properties. Some of the
classical schemes include filter generators, combiner generators, clock controlled
generators and shrinking generators. The nonlinear function used in these generators is a boolean function $f\; :  \; GF(2^n) \rightarrow GF(2)$ which takes $l$ inputs from either an LFSR or from outputs of number of LFSRs and produces $ GF (2)$ keystream sequence. LFSR based sequence generators serve as the basic building blocks for a number of e-stream submissions~\cite{estream} and cellular algorithms i.e. A5
series~\cite{cellular} and E0 algorithm~\cite{bluetoothspecification}.

In many scientific applications, it is common to analyze a problem in different domains because some characteristics of the involved parameters can only be better revealed in a particular domain. Moreover, the computational requirements can be reduced as result of the analysis in some transformed domains. For instance, the prorogation
characteristics of light or magnetic waves can be conveniently observed in frequency domain as compared to its equilvalent representation in time domain. A continous signal in time domain   $s(t) = cos (2\pi f_ct)$ having infinite non-zero points once converted to frequency domain $S(f)$ has only two non-zero frequency components i.e. $f_c$ and $-f_c$. Similarly decoding of Block codes like  Bose, Chaudhuri, and Hocquenghem (BCH) codes and Reed Solomon codes in frequency domain have proven to be more efficient in terms of computational complexity and error free recovery as compared to its analogous time
domain decoding techniques~\cite{peter2002error}. Contemplating the same, analysis of cryptographic primitives in transform domains has also produced promising results. Though very less known in public literature, DFT based spectral attacks ~\cite{gong2011fast} and transform domain analysis of DES cipher components ~\cite{gong1999transform} are interesting examples in this regard.
 
In this report, analysis of LFSR based sequence generators in time and frequency domains has been presented. Starting form time and frequency domain analysis of basic LFSR sequences, we build our analysis onto filter and combiner generators. In Section-2, basics of Fourier transform over finite fields is recalled. Section-3 delineates the time and frequency domain analysis of LFSRs. Section-4 describes transform domain analysis of a simple product sequence which is fundamental component of non-linear boolean functions. A novel account of Chinese Remainder Theorem (CRT) based interpretation of fixed patterns in cyclic structures of underlying finite fields is discussed in this section. In Section-5, time and frequency domain analysis of non-linear filter generators is given with a perspective of their application in cryptographic algorithms. Specific comments about selective DFT attacks on filter generators~\cite{gong2011fast} are specifically made in this section. Section-6 discusses transform domain analysis of combiner generators and application of selective DFT attacks on combiner generators by exploiting modular computations of CRT alongwith detailed complexity comparison in relation to classical divide and conquer attacks. In Section-7 discussion on applicability of discrete fourier spectra attacks on A5/1 algorithm is made. The report is finally concluded in section-8. 
\section{Frequency Domain Representation over Finite Fields}
Discrete Fourier Transform (DFT) is considered one of the most important discovery in the area of  signal processing. DFT presents us with an alternate mathematical tool that allows us to examine the frequency domain behaviour of signals, often revealing important information not apparent in time domain. DFT $S_k$ of an n-point sequence $s_i$ is expressed in terms of inner product between the sequence and set of complex discrete frequency exponentials:

\begin{equation}\label{DFT-dsp eq}
  S_{k} = \sum_{i=0}^{n-1} s_{i}e^{-j2\pi ik/n} , \;\;\;  k = 0,1,2,.....,n-1
\end{equation}
 The term $e^{-j2\pi ik/n}$ represents discrete set of exponentials. Alternatively, $e^{-j2\pi/n}$ can be viewed as $n^{th}$ root of unity.  


Analogous to the classical DFT, a DFT for a periodic signal $s_t$ with period $n$ defined over a finite field $GF(2^m) $  is represented as

\begin{equation}\label{DFT eq}
  S_{k} = \sum_{t=0}^{n-1} s_{t}\alpha^{-tk} , \;\;\;  k = 0,1,2,.....,n-1
\end{equation}
 where $S_k$ is $k$-th frequency component of DFT and $\alpha$ is the primitive element; generator of $GF(2^m) $ with period $n$ ~\cite{pollard1971fast}. For Inverse DFT, we will have a relation 
\begin{equation}\label{IDFT eq}
  s_t = \sum_{k=0}^{n-1} S_{k}\alpha^{tk} , \;\;\;  k = 0,1,2,.....,n-1
\end{equation}

Similarly for polynomials, we have a relation for DFT and IDFT. Having a correspondence between a minimum polynomial and its associated sequence $s_t$ with $s(x) = \sum_{t=0}^{n-1}s_{t} x^{t}$ and $S(x) = \sum_{k=0}^{n-1}S_{k} x^{k}$, following relation holds for DFT~\cite{golomb2005signal}: 
\begin{equation}\label{DFT eq-poly}
  S_{k} = s(\alpha^{-k}), \;\;\;  k = 0,1,2,.....,n-1
\end{equation}
and similarly for IDFT:
\begin{equation}\label{DFT eq-poly}
  s_t = S(\alpha^{t}), \;\;\;  t = 0,1,2,.....,n-1
\end{equation}

\section{Transformed Domain Analysis of LFSR Sequences}
Classical theory on LFSR sequences and their applications in cryptology can be found in ~\cite{golomb1982shift}, ~\cite{golomb2005signal} and ~\cite{rueppel1986analysis}. In this section, transformed domain analysis of LFSRs, their sequences and underlying algebraic structures are recalled as they are fundamental to our proposed approach on filter and combiner generators.

\subsection{Time Domain Representation of an LFSR Sequence} 
 A binary sequence $s_t$ can be treated as an LFSR sequence of degree $m$ if it follows a linear recursion with coefficents from $GF(2)$ as: 
\begin{equation}\label{eq:recur}
s_{i+m} = \sum_{k=0}^{m-1} c_{k}s_{i+k}\;\;\;\;\;\mbox{for}\;\;i \geq 0
\end{equation}
The value $m$ is called the order of recurrence and associated characteristic polynomial in $GF(2)$ is defined by
\begin{equation}\label{eq:charpoly}
f(x) = x^m + \sum_{k=0}^{m-1} c_{k}s_{i+k}\;\;\;\;\;\mbox{for}\;\;i \geq 0
\end{equation}

The initial state $(s_0, s_1, .... , s_{m-1})$ of an LFSR serves as a key to generate the complete sequence $s_t$ . The period of any non-zero sequence can be utmost $2^m - 1$ which is in relation to the characteristic polynomial $f(x)$ of the LFSR. If $f(x)$ is irreducible, it has $m$ distint roots i.e. $\alpha$ and its conjugate set $\{\alpha, \alpha^2, \alpha^4,..., \alpha^{2^{m-1}}\}$. Consequently, if $f(x)$ is a primitive polynomial, then order of its root $\alpha$ must be $2^m - 1$ which in other words is a period of associated sequence $s_{t}$. Thus a sequence $s_t$ of an LFSR given by a primitive polynomial has maximum possible period $2^m - 1$ and is called $m$-sequence. \\

\textbf{Trace Representation of an LFSR Sequence}.\;\;\; 
The same sequence $s_t$ can also be expressed in terms of its trace representation~\cite{ronjom2007attacks}; a linear operator from $GF(2^m)$ to its subfiled $GF(2)$ . Let $Tr_{1}^{m}(x) = \sum_{k=0}^{m-1} {x^{2}}^{k}$ be the trace mapping from $GF(2^m)$ to $GF(2)$, then $m$ sequence $s_t$ can be represented as:
\begin{equation}\label{eq:trace}
s_t = Tr_{1}^{m} (\beta \alpha^t)
\end{equation}
\par where $\alpha$ is a generator of a cyclic group $GF(2^{m})^*$ and is called as primitive element of $GF(2^{m})$. Note that $\beta \in GF(2^{m})$ and each of its nonzero value corresponds to cyclic shift of the $m$-sequence generated by an LFSR with primitive polynomial $f(x)$. Importance of this interpretation of $m$-sequence is that different sequences constructed from root $\alpha$ of primitive polynomial $f(x)$ are cyclic shifts of the same $m$-sequence. The associated linear space $G(f)$ of dimension $m$ contains $2^{m}$ different binary sequences including all 0s sequence as:
 \begin{equation}\label{eq:gf}
G(f) = \{\; \tau^{i}s \; \vert \; 0\; \leq \;i \;\leq \;2^m - 2\; \} \bigcup \{0\}   
\end{equation}
 where $\tau$ is a left shift operator and represents a linear transformation of sequence $s_t$. It is important to mention here that all sequences in $G(f)$, defined over a primitive polynomial $f(x)$, have maximum period $r$ i.e. $2^{m}-1$ with an obvious exception of all 0s sequence. Moreover, any two sequences $s$ and $u$ within $G(f)$ are cyclic shift equilvalent, if there exists an integer $k$ such that
\begin{equation}\label{eq:shift}
u_i = s_{i+k},\;\;\;\;\;\;\; \forall \; i \; \geq 0.
\end{equation}

\textbf{LFSR Sequence in Matrix Form}.\;\;
Each state of an $m$ stage LFSR is a vector in the $m$-dimensional space $GF(2^{m})$. The shift register is then a linear operator that changes the current state to its successor vector according to the feedback. In simple terms, the transformation of each non-zero sequence in a field, from state
$ (s_{k},s_{k+1}, ....,s_{k+m-1})$  to its successor state
 $ (s_{k+1},s_{k+2}, ....,s_{k+m})$
 can be regarded as a linear operation on  $GF(2^{m})$. The advantage of working with operator operating on $m-$dimensional vector space is that it affords a matrix representation. Since, 
\begin{equation}\label{eq:mat-LFSR}
 s_{m+k} = c_{0}s_{k}+ c_{1}s_{k+1}+...+c_{m-1}s_{k+m-1} , k \geq 0.
\end{equation}

 hence a shift register matrix takes the form:

\[ T = \left[ \begin{array}{cccccc}
0 & 0 & 0 & ... & 0 & c_{0} \\
1 & 0 & 0 & ... & 0 & c_{1} \\
0 & 1 & 0 & ... & 0 & c_{2}\\
. & . & . & ... & . & . \\
. & . & . & ... & . & . \\
. & . & . & ... & . & . \\
0 & 0 & 0 & ... & 1 & c_{m-1} \end{array} \right]\]

 and

\[ \begin{array} {lcl} (s_{k+1},s_{k+2}, ....,s_{k+m}) & =  & (s_{k},s_{k+1}, ....,s_{k+m-1})T \\
 &= & (s_{k-1},s_{k}, ....,s_{k-1+m-1})T^{2} \\
& = & ...  \\
& = & (s_{0},s_{1}, ....,s_{m-1})T^{k+1}
 \end{array} \]

The matrix T is called the state matrix of the LFSR in time domain. Note that
 $det(T) = (-1)^{m}c_{0} $. Thus T is invertible if and only if $c_{0} \neq 0 $.

\subsection{Frequency Domain Representation of an LFSR Sequence}
 From \eqref{DFT eq}, DFT of an LFSR sequence $a_t$ defined over a primitive polynomial $f(x)$ produces a Fourier spectrum sequence $A_k$. For completeness, few important facts are reproduced here from ~\cite{Blahut1983errorcontrol} and ~\cite{massey1994fourier} with some novel observations as well:
 \begin{enumerate}
   \item The zero components in the Fourier spectrum of a sequence over $GF(2^m)$ are related to the roots of a polynomial of that sequence. For example, DFT of an LFSR sequence with feedback polynomial $f(x) = x^3+x+1$ initialized with state $001$ is $0, 0, 0, \alpha^4 , 0, \alpha^2 , \alpha$. As roots of $f(x)$ are $\alpha$ alongwith its conjugates i.e. $\alpha^2$ and $\alpha^4$, so first, second and fourth spectral components are zero.
   \item As DFT of a time domain signal comprises of a fundamental frequency and its harmonics, DFT of an LFSR sequence based on a minimal polynomial with no multiple roots also comprises of $\alpha^i \in GF(2^m)$  and its harmonics $\alpha^{i{^j} mod\mbox{\;} r} \in GF(2^m)$ with $0\leq i \leq r-1$. This harmonic pattern can be efficiently exploited in cryptanalysis attacks on LFSR based sequence generators.
	\item All DFT components of an LFSR sequence $\in GF(2^m)$.   
   \item indices of non zero DFT points for LFSR with minimum polynomial and no multiple roots also follow a fixed pattern. If $k$-th component of spectral sequence is non zero then all $(2^{j}k)\; mod\; r$ components will be harmonics of the $k$-th component where  $1\leq j \leq m-1$.
   
   \item DFT of an LFSR sequence based on a polynomial with multiple roots does not contain harmonic pattern of elements.
   
   \item The Linear Complexity $L$ of an n-periodic sequence is equal to Hamming weight of its frequency domain associate. Three non zero spectral components in example above verifies this fact.
    
	\item \textbf{Time Shift Property}.  Let two sequences related by a time shift $u_{t} = s_{t+\tau}$, their DFTs $U_{k}$ and $S_{k}$ are related as:
\begin{equation}\label{DFT shift eq}
  U_{k} = \alpha^{k\tau} S_{k}, \;\;\;\;\;   k = 0,1,....,n-1
\end{equation}
\item indices of non-zero spectral points of an LFSR sequence does not change with the shift in LFSR sequence. A non-zero $k$-th component of DFT of an LFSR sequence will always be non-zero. Any shift in LFSR sequence will only change the value at this component by~(\ref{DFT shift eq}). Converse is also true for zero spectral points of an LFSR sequence which will always be zero no matter how much sequence is shifted.
\item \textbf{Trace Representation of an LFSR sequence in Frequency Domain}.\;\;
A binary sequence $s_t$ can be represented in terms of trace function with spectral componenets as follows:-

\begin{equation}\label{eq:trace in DFT}
s_t = \sum_{j\in \Gamma(n)} Tr_{1}^{m_j} (A_j \alpha^{-jt}),\;\;\;\;\;t = 0,1,....,n-1
\end{equation}
 where $Tr_{1}^{m_j}$ is a trace function from $F_{2^m}$ to $F_2$, $A_j\in F_{2^m}$ and $\Gamma(n)$ is a set of  cyclotomic coset leaders modulo $n$. 
\item \textbf{Matrix Representation in Frequency Domain}.\;\;
DFT of $s_t$, being a linear operator with respect to $\alpha\in GF(2^{m})$ from equation \ref{DFT eq}, can be written in matrix form as:
 \begin{equation}\label{DFT Matrix}
   (S_{0}, S_{1}, S_{2},..., S_{n-1})^T = D (s_{0}, s_{1}, s_{2},..., s_{n-1})^T
 \end{equation}
 where
  \[ D = \left[ \begin{array}{ccccc}
1 & 1 & 1 & ... & 1  \\
1 & \alpha & \alpha^{2} & ... & \alpha^{n-1}  \\
1 & \alpha^{2} & \alpha^{2.2} & ... & \alpha^{2(n-1)} \\
. & . & . & ... & .  \\
. & . & . & ... & .  \\
. & . & . & ... & .  \\
1 & \alpha^{n-1} & \alpha^{2(n-1)} & ... & \alpha^{n-1(n-1)}  \end{array} \right]\]
\end{enumerate}

\section{Transform Domain Analysis of a Product Sequence}

In this section, transform domain analysis of a product sequence generated through multiplication of two LFSR sequences is presented. This analysis forms the basis for the combiner and filter generators that will appear in subsequent sections. The spectral domain features discussed in Section-3.2 hold true for product sequences as well. In addition, a linear structure existing in the frequency domain representation of the product sequence is presented which renders itself useful for cryptanalysis of LFSR based sequence generators.
The initial state of an LFSR has direct relevance to the element $\beta \in \; GF(2^m)$ in (\ref{eq:trace}) which has been extensively exploited in algebraic and DFT based spectra attacks~\cite{gong2011fast}. Akin to this, there exists another phenemenon which has one to one correspondence with initial states of LFSRs within a linear space $G(f)$ containing cyclic shifts of $m$ sequences. These cyclic shifts in LFSRs sequences and their correspondence to maximum period $r$ posseses certain fixed patterns which exhibit linear behaviour even when employed in nonlinear combiner generators. We have observed that CRT interprets the shifts in LFSRs sequences  and is considered as our major contribution to classical theory of LFSRs and sequence generators. As the product of two LFSR sequences is a building block of any non linear boolean function, the idea has been discussed by considering a simple case of two $m$ sequences multiplied togather. The process has been generalized through mathematical rationale later in this subsection where CRT based interpretaion of shifts in LFSRs sequences has been discussed. 
\subsection{Time Domain Analysis of a Product Sequence}
 We build our analysis by starting with a simple case of  multiplication of output sequences of two LFSRs and illustrate  our novel observations on fixed structures existing in the frequency domian representation of product sequences. The observations of this special case will be generalized to a sequence generator in the next section.

 Let $s_{t}$ be a key stream generated by multiplying the two LFSRs sequences $a_{t}$ and $b_{t}$ defined as
 \begin{equation}\label{eq:twoLFSR1}
s_t = f(a_t,b_t)
\end{equation}
where $ f(.)$ is a nonlinear function representing a term wise product. If period of $a_t$ is $r_1$ and $b_t$ is $r_2$, we have 
 \begin{equation}\label{eq:twoLFSR}
s_t =   a_i\;.\; b_i\;\;\;\;\;\;\;\; \mbox{with} \;\; 0\leq \; i \; < r\;
\end{equation} 
where $r\; =\;\mbox{lcm}\;(r_1,r_2) $. The linear complexity $L$ of the sequence $s_t $ in this case satisfies 
\begin{equation}\label{eq:LC-product}
L (s_t) \leq L(a_t)L(b_t)
\end{equation}
where $L$ denotes linear complexity of a sequence and the equality in (\ref{eq:LC-product}) holds only iff associated polynomials of $a_t$ and $b_t$ are primitive and are greater than 2~\cite{golomb2005signal}.

\subsubsection{Fixed Patterns in Cyclic Structures of LFSRs}
It has been observed that there exists a specific relationship between the amount in shifts of product sequence $s_t$, period of individual LFSRs and shifts from their refernce initial states. We generalize the process by giving its mathematical rationale followed by detailed discussion through a small example. CRT allows mathematical representation of relationship observed between the shifts in individual LFSRs, their corresponding periods and overall impact on product sequence $s_t$. We have a important theorem here.
\begin{theorem}

\label{CRT-theorem}
Let $s_t\in GF(2^m)$ be a reference product sequence with period $n$ having two constituent LFSRs defined over primitive polynomials with individual periods $n_1$ and $n_2$. With different shifts $k_1$ and $k_2$ in initials states of LFSRs, resulting output sequnece $u_t$ is correlated to $s_t$ by~(\ref{eq:shift}) where shift $\tau$ is determined through CRT as

      \begin{eqnarray*}
      \tau  &\equiv& k_1\mbox{\; (mod} \mbox{\;}n_1) \\
      \tau  &\equiv& k_2\mbox{\; (mod} \mbox{\;}n_2)   
               \end{eqnarray*}

\end{theorem}
\begin{proof}
Within a cyclic group $GF(2^m)$, associated linear space $G(f)$ of dimension $m$ contains $2^{m}-1$ non-zero binary sequences by~(\ref{eq:gf}).
\\ As $s_t$ and $u_t$ both $\in$ $GF(2^m)$, they are shift equilvalents by~(\ref{eq:shift}) with unknown shift value of $\tau$. 
\\ The product sequnec $s_t$ of $a_t$ and $b_t$ can be expressed as
\begin{equation}
s_i = a_{j}.b_{v} 
\end{equation}

 where  $0\leq i \leq n-1$ , $0\leq j \leq n_{1}-1$ and   $0\leq v \leq n_{2}-1$. 
\begin{remark}
\label{C-1}
While contributing towards a product sequence of length $n$ with two LFSRs, stream of LFSR-1 defined over $GF(2^p)$ with primitive polynomial and its maximum period $2^{p}-1$  is repeated $\delta_1$ times while LFSR-2 defined over $GF(2^q)$ with primitive polynomial as well and corresponding period $2^{q}-1$ is repeated $\delta_2$ where

\begin{eqnarray*}
      \delta_1  &= & \frac{\mbox{lcm}(n_1,n_2)}{n_1}\mbox{\;,\;\; and\;}  \\
      \delta_2  &= & \frac{\mbox{lcm}(n_1,n_2)}{n_2}   
               \end{eqnarray*} 
               \end{remark}

\begin{remark}
\label{C-2}
 Within a sequence of period $n$ for a product sequence, each value of index $j$ corresponds to all values of index $v$ if and only if $gcd(n_1,n_2)=1$.
\end{remark}
From Remarks~\ref{C-1} and~\ref{C-2}, any shift in LFSRs initial states will produce output corresponding to some fixed indices of $j$ and $v$ which already existed in the refernce sequence at some fixed place with initial states of LFSRs without shift.\\
\\ With known values of $j$ and $v$ i.e. $k_{i's}$, CRT will give us the value of $\tau$ mod $n$ as
\begin{eqnarray*}
      \tau  &\equiv& k_1\mbox{\; (mod} \mbox{\;}n_1) \\
     \tau  &\equiv& k_2\mbox{\; (mod} \mbox{\;}n_2)   
               \end{eqnarray*}
\end{proof}	
\begin{example}
Let we have a sequence $s_t$ generated from product of two LFSRs having primitive p[olynomials of $g_1(x)  x^2+x+1$ and $g_2(x)  x^3+x+1$. The period $n_1$ of stream $a_t$ corresponding to LFSR-1 is $3$ and $n_2$ of $b_t$ corresponding vto LFSR-2 is $7$. The period $n$ of $s_t$ is $21$.

Table~\ref{tab:cyclic-m} demonstrates product of two $m$ sequences generated from these two LFSRs. 
 \begin{table}[H]
\small
\begin{center}
\caption[Sample Table]{Product sequence of 2x LFSRs with $n_1 = 3$ and $n_2 = 7$}
\begin{tabular}{|c|c| c| c| c| c| c| c| a| c| c| c| c| c| c| c| a| c| c| c| c| c| }
\hline
 Shift Index &0 & 1 & 2 & 3 & 4 & 5 & 6 & 7 & 8 & 9 & 10 & 11 & 12 & 13 & 14 & 15 & 16 & 17 & 18 & 19 & 20  \\ \hline
  
 $a_t$&$a_{1}$ & $a_{2}$ &$a_{3}$ & $a_{1}$ & $a_{2}$ & $a_{3}$ & $a_{1}$ & $a_{2}$ &$a_{3}$ & $a_{1}$ & $a_{2}$ & $a_{3}$ & $a_{1}$ & $a_{2}$ & $a_{3}$ & $a_{1}$ & $a_{2}$ & $a_{3}$ & $a_{1}$ & $a_{2}$ & $a_{3}$  \\ \hline

$b_t$& $b_{1}$ & $b_{2}$ &$b_{3}$ & $b_{4}$ & $b_{5}$ & $b_{6}$ & $b_{7}$ &  $b_{1}$ & $b_{2}$ &$b_{3}$ & $b_{4}$ & $b_{5}$ & $b_{6}$ & $b_{7}$ & $b_{1}$ & $b_{2}$ &$b_{3}$ & $b_{4}$ & $b_{5}$ & $b_{6}$ & $b_{7}$  \\ \hline

 $s_t$&$s_{1}$ & $s_{2}$ &$s_{3}$ & $s_{4}$ & $s_{5}$ & $s_{6}$ & $s_{7}$ &  $s_{8}$ & $s_{9}$ &$s_{10}$ & $s_{11}$ & $s_{12}$ & $s_{13}$ & $s_{14}$ & $s_{15}$ & $s_{16}$ &$s_{17}$ & $s_{18}$ & $s_{19}$ & $s_{20}$ & $s_{21}$  \\ \hline
\end{tabular}

\label{tab:cyclic-m}
\end{center}

\end{table}

We analyze the impact of shift on LFSR sequences and their behaviour in cyclic stuctures of finite fields involved. We will shift the LFSR sequences one by one and observe the fixed patterns which can be exploited in cryptanalysis of the combiner generators. We can represent shifts in LFSRs sequences with $k$ and $l$ as 

 \begin{equation}\label{eq:twoLFSR-shift}
s_t =   a_{i+k}.b_{i+l}\;,\;\;\;\;\;\;\;\mbox{with} \; 0\leq \; i \; \leq n-1
\end{equation}

where $k \in [0,n_{1}-1]$ and $l \in [0,n_{2}-1]$. Table~\ref{tab:1-bit-a} demonstrates the scenerio where $a_t$ is left shifted by one bit while keeping the $b_t$ fixed with initial state of '1'.

\begin{table}[H]
\small
\begin{center}
\caption[Sample Table]{Product sequence with $a_{t}$ shifted left}
\begin{tabular}{|c|a| c| c| c| c| c| c|a| c| c| c| c| c| c| a| c| c| c| c| c| c| }
\hline

 Shift Index &7 & 8 & 9 & 10 & 11 & 12 & 13 & 14 & 15 & 16 & 17 & 18 & 19 & 20 & 0 & 1 & 2 & 3 & 4 & 5 & 6  \\ \hline
  
 $a_t$ & $a_{2}$ &$a_{3}$ & $a_{1}$ & $a_{2}$ & $a_{3}$ & $a_{1}$ & $a_{2}$ &$a_{3}$ & $a_{1}$ & $a_{2}$ & $a_{3}$ & $a_{1}$ & $a_{2}$ & $a_{3}$ & $a_{1}$ & $a_{2}$ & $a_{3}$ & $a_{1}$ & $a_{2}$ & $a_{3}$ &  $a_{1}$  \\ \hline

$b_t$ & $b_{1}$ & $b_{2}$ &$b_{3}$ & $b_{4}$ & $b_{5}$ & $b_{6}$ & $b_{7}$ &  $b_{1}$ & $b_{2}$ &$b_{3}$ & $b_{4}$ & $b_{5}$ & $b_{6}$ & $b_{7}$ & $b_{1}$ & $b_{2}$ &$b_{3}$ & $b_{4}$ & $b_{5}$ & $b_{6}$ & $b_{7}$  \\ \hline
 
 $u_t$&$s_{8}$ & $s_{9}$ &$s_{10}$ & $s_{11}$ & $s_{12}$ & $s_{13}$ & $s_{14}$ &  $s_{15}$ & $s_{16}$ &$s_{17}$ & $s_{18}$ & $s_{19}$ & $s_{20}$ & $s_{21}$ & $s_{1}$ & $s_{2}$ &$s_{3}$ & $s_{4}$ & $s_{5}$ & $s_{6}$ & $s_{7}$  \\ \hline
\end{tabular}

\label{tab:1-bit-a}
\end{center}
\end{table}

Comparison of Table~\ref{tab:cyclic-m} with Table~\ref{tab:1-bit-a} reveals that shifting one bit left of $a_t$ and fixing the $b_t$ to reference initial state of '1' shifts $s_t$ by seven units left. Similarly, shifting another bit of $a_t$ to left, brings $a_{3}$ corresponding to $b_{1}$ which can be located in Table~\ref{tab:cyclic-m} at shift position 14. So two left shifts of $a_t$ shifts $s_t$ by 14 units left with reference to bit positions in Table~\ref{tab:cyclic-m}. Now we analyze the impact of left shift of $b_t$ on $s_t$. Table~\ref{tab:1-bit-b} demonstrates the scenerio where $b_t$ is left shifted by one bit while keeping the $a_t$ fixed with initial state of '1'. 
\begin{table}[H]
\small
\caption[Sample Table]{Product sequence with $b_{t}$ shifted left}
\begin{center}
\begin{tabular}{|c|a| c| c| c| c| c| a| c| c| a| c| c| c| c| c| a| c| c| c| c| c| }
\hline
 Shift Index&15 & 16 & 17 & 18 & 19 &20 &0 & 1 &2 & 3 & 4 & 5 & 6 & 7 & 8 & 9 & 10 & 11 & 12 & 13 & 14 \\ \hline
 $a_t$&  $a_{1}$ &  $a_{2}$& $a_{3}$ & $a_{1}$ & $a_{2}$ & $a_{3}$ & $a_{1}$ & $a_{2}$ &$a_{3}$ & $a_{1}$ & $a_{2}$ & $a_{3}$ & $a_{1}$ & $a_{2}$ & $a_{3}$ & $a_{1}$ & $a_{2}$ & $a_{3}$ & $a_{1}$ & $a_{2}$ & $a_{3}$   \\ \hline

  $b_t$&$b_{2}$ &$b_{3}$ & $b_{4}$ & $b_{5}$ & $b_{6}$ & $b_{7}$ &  $b_{1}$ & $b_{2}$ &$b_{3}$ & $b_{4}$ & $b_{5}$ & $b_{6}$ & $b_{7}$ & $b_{1}$ & $b_{2}$ &$b_{3}$ & $b_{4}$ & $b_{5}$ & $b_{6}$ & $b_{7}$ & $b_{1}$  \\ \hline
 
 $u_t$&$s_{16}$ & $s_{17}$ &$s_{18}$ & $s_{19}$ & $s_{20}$ & $s_{21}$ & $s_{1}$ &  $s_{2}$ & $s_{3}$ &$s_{4}$ & $s_{5}$ & $s_{6}$ & $s_{7}$ & $s_{8}$ & $s_{9}$ & $s_{10}$ &$s_{11}$ & $s_{12}$ & $s_{13}$ & $s_{14}$ & $s_{15}$  \\ \hline

\end{tabular}

\label{tab:1-bit-b}
\end{center}
\end{table}

It can be easily seen that one left shift in $b_t$ shifts  $s_t$ by 15 units where $b_{2}$ is corresponding to $a_{1}$. Similarly, another left shift in $b_t$ shifts $s_t$ by another 15 units bringing the $b_{3}$ corresponding to $a_{1}$. Subsequently, three left shifts in $b_t$ with reference to initial state of '1' brings $b_{4}$ corresponding to $a_{1}$ which is at shift index-3 in Table~\ref{tab:cyclic-m}. Similar fixed patterns can be observed for simultaneous shifts of LFSRs and it will be discussed with more detail in following paragraphs.

Let us model this fixed patterns in LFSRs cyclic structures and shifts in intial states of LFSRs through CRT as

	\begin{eqnarray*}
      x  &\equiv& k\mbox{\; (mod} \mbox{\;}n_1) \\
      x  &\equiv& l\mbox{\; (mod} \mbox{\;}n_2)   
               \end{eqnarray*}
	where $k$ and $l$ denote the amount of shifts in initial state of individual LFSRs with reference to initial state of '1'. The solution of CRT i.e. $x$(mod $r$) gives the amount of shift in $s_t$ with reference to $u_t$ as depicted in (\ref{eq:shift}). Consider a scenerio again where $a_t$ is shifted left by one bit and $b_t$ is fixed with initial state of '1' and can be expressed as 
	\begin{eqnarray*}
      x  &\equiv& 1\mbox{\; (mod} \mbox{\;}3) \\
      x  &\equiv& 0\mbox{\; (mod} \mbox{\;}7)   
               \end{eqnarray*}	
	The CRT gives the solution of 7(mod $21$) which is index position of $a_2$ corresponding to $b_1$ in Table~\ref{tab:cyclic-m} shifting the product sequence $s_t$ by seven units left. Consider another scenerio of simultaneous shifts in both LFSRs sequences where $a_t$ is shifted left by one bit and $b_t$ is shifted left by 3 bits with reference to their initial states of '1' and can be expressed as
	\begin{eqnarray*}
      x  &\equiv& 1 \mbox{\; (mod} \mbox{\;}3) \\
      x  &\equiv& 3 \mbox{\; (mod} \mbox{\;}7)   
               \end{eqnarray*}
	The CRT gives value of $-11$ which is $10$ (mod $21$), representing the product sequence $u_t$ as 10 units left shifted version of $s_t$. This value matches to index position of $b_{4}$ corersponding to $a_{2}$ in Table~\ref{tab:cyclic-m}.

 Our Observations related to direct correspondence of shift index with initial states of LFSRs and CRT calculations done modulo periods of individual LFSRs are valid for any number of LFSRs in different configurations of nonlinear sequence generators. These observations on classical theory of LFSR cyclic structures with their CRT based interpretation are considered significant for cryptanalysis. 

 \end{example}
 
\subsection{Frequency Domain Analysis of a Product Sequence}
To compute the DFT of the sequence $s_t$ using equation~(\ref{DFT eq}), we need to know the minimium polynomial for $s_t$ $\in GF(2^m)$ which can be efficiently determined through Berlekamp-Massey algorithm. It was demonstrated in last subsection that a linear structure exists in the spectral representation of component sequences $a_t$ and $b_t$ which propagates further in the DFT spectra of product sequence $s_t$. Few interesting results are presented here duly illustrated by an example:
 \begin{enumerate}
    
\item  Zero and non-zero positions of the DFT Spectra of $s_t$ can be determined even without knowing the minimum polynomial for $s_t \in GF(2^m)$ while working in the lower order associated fields of $a_t \in GF(2^p)$ and $b_t \in GF(2^q)$, where $p \mbox{\;} \& \mbox{\;} q  \mbox{\;}\textless \mbox{\;} m$. Any $k$-th component of DFT spectra of $s_t$ is non-zero if and only if $A_{k}$ and $B_{k}$ are both non zero, where $A_k$ and $B_k$ represents DFTs of $a_t$ and $b_t$ respectively.
\item  With known non-zero DFT points for $a_t$ and $b_t$, Chinese Remainder Theorem (CRT) can be used to determine non zero  points of DFT spectra of $s_t$ directly. With $n_1$ and $n_2$ be the individual periods for $a_t$ and $b_t$ respectively, we can apply CRT as:
   
   \begin{eqnarray*}
      x  &\equiv& k_1\mbox{\; (mod} \mbox{\;}n_1) \\
      x  &\equiv& k_2\mbox{\; (mod} \mbox{\;}n_2)   
               \end{eqnarray*}
   where $k_1$ and $k_2$ are non zero index positions of $A_k$ and $B_k$ respectively and $x$ is the position of non-zero componenet of DFT spectra of $s_t$ within its period $n$.
   
   \item DFT of a product sequence with minimal polynomial having no multiple roots follows a harmonic structure of its elements with $\alpha^i \in GF(2^m)$  and its harmonics $\alpha^{i{^j} \mbox{mod\;} n} \in GF(2^m)$ with $0\leq i \leq n-1$, appearing in DFT spectrum. 
	\item Non-zero indices of DFT sequences also follow a fixed pattern. A k-th non-zero component has its harmonics at all $(2^{j}k)\; \mbox{mod\;} n$ points with $1\leq j \leq n-1$. 
	\item Shifting of any component sequence  $a_t$ or $b_t$ will impact the spectral components of resulting sequence $s_t$ by (\ref{DFT shift eq}).

\item The zero components in the fourier transform of a product sequence $s_t$ defined over over $GF(2^m)$ are related to roots of its minimum polynomial $g_(x)$.
 
\item Consider two LFSR sequences $a_t \in GF(2^p)$ and $b_t \in GF(2^q)$ being components of a product sequence, their spectral components as $A_k$ and $B_k$ respectively with $0 \leq k \leq N-1$. While working in base fields, can we determine their corresponding frequency components of product stream $z_t \in GF(2^m)$?. Our results on this problem are being published somewhere else shortly.
  
   \end{enumerate}

\begin{example}
Consider a product sequence $s_t$ generated from two LFSRs with minimum polynomials  $g_{1}(x) = x^3 + x +1$ and  $g_{2}(x) = x^2 +x +1$.

\begin{enumerate}
\item In time domain representation, we have following sequences.\\
\;\;\;\;\;\;\;\;Sequence $a_t$:\;\;\;$011$\;\;\;\;\;\;\;\;\;\;\;\;\;\;\;\;\;\;\;\;\;\;\;\;\;\;\;\;\;\;\;\;\;\;\;\;\;\;\;\;\;\;\;\;\;\;\;\;\;\;\;\;\;\;\;\;\;\;\;\;\;\;\;\;\;\;\;\;\;\;\;\;\;(of period 3)\\
\;\;\;\;\;\;\;\;	Sequence $b_t$:\;\;\;$0010111$\;\;\;\;\;\;\;\;\;\;\;\;\;\;\;\;\;\;\;\;\;\;\;\;\;\;\;\;\;\;\;\;\;\;\;\;\;\; \;\;\;\;\;\;\;\;\;\;\;\;\;\;\;\;\;\;\;\;\;\;\;\;\;\;(of period 7)\\
\;\;\;\;\;\;\;\;	Sequence $s_t$:\;\;\;$001011000001010010011101110111011101$\;\;\;\;\;\;\;\;\;\;\;\;\;\;(of period 21)\\
\item From~(\ref{DFT eq}), frequency domain representations of these sequences are:
\begin{enumerate}
\item $A_k = 0,1,1$
\item $B_k = 0,0,0,\alpha^4, 0,\alpha^2,\alpha$
\item To compute $S_k$, associated minimum polynomial is determined through Berlekamp-Massey algorithm which is 
$g(x) = x^6+x^4+x^2+x+1$.\\
$S_k = 0,0,0,0,0,\alpha^9,0,0,0,0, \alpha^{18},0,0,\alpha^{15},0,0,0,\alpha^{18},0,\alpha^9,\alpha^{15}$
\end{enumerate} 

\item Following is to be notified here:-

\begin{enumerate}
\item Non-zero DFT points in $S_k$ clearly follow a linear behaviour as of time domain representation where any $k$-th component is non-zero if and only if $A_{k}$ and $B_{k}$ are both non-zero. CRT can be directly used to determine these non-zero points.
\item Harmonic pattern of DFT spectra are visible for $A_k$, $B_k$ and $S_k$. 
\item Non-zero indices of DFT sequences also follow a fixed pattern. In case of $S_k$, non zero DFT element at index 5 has its harmonics at indices $10,\; 20,\; 19 \;(40\; \mbox{mod\;} 21), \;17\;(80\; \mbox{mod\;} 21)$ and at $13\;(160\; \mbox{mod\;} 21)$.
\item If we shift $b_t$ by one bit to left, resulting sequences in frequency domain will hold shift property of equation (\ref{DFT shift eq}).\\
$A_k = 0,1,1$ \\
 $B_k = 0,0,0,\alpha, 0,\alpha^4,\alpha^2$ \\
$S_k = 0,0,0,0,0,\alpha^{18},0,0,0,0, \alpha^{15},0,0,\alpha^{9},0,0,0,\alpha^{15},0,\alpha^{18},\alpha^{9}$
\item The zero components in the fourier transform of a product sequence $s_t$ defined over $GF(2^m)$ are related to roots of $g_(x) = x^6+x^4+x^2+x+1$. As roots of $g(x)$ are $\alpha$ alongwith its conjugates i.e. $\alpha^2$, $\alpha^4$, $\alpha^8$ and $\alpha^{16}$ so first, second, fourth, eigth and sixteenth spectral components are zero.

\end{enumerate}

\end{enumerate}
\end{example}


%
The patterns observed in during time and frequency domain analysis of product sequences will be applied to LFSR based sequence generators in the following sections. LFSR based sequence generators can be broadly divided in to three main classes; non-linear filter generators, non-linear combiner generators and clock controlled generators with few variants of shrinking generators~\cite{canteaut2011stream}.  Transform domain anlysis of filter and combiner generators will be presented in following sections.


\section{Transformed Domain Analysis of Filter Generators}
The nonlinear filter generator consists of a single LFSR which is filtered by a nonlinear boolean function $f$ and is called as filtering function. In a filter generator, the LFSR feedback polynomial, the filtering function and the tapping sequence are usually publicly known. The secret parameter is the initial state of the LFSR which is derived from the secret key of the cipher by a key-loading algorithm. Therefore, most
attacks on filter generators consist of recovering the LFSR initial state from the knowledge of some bits of the sequence produced by the generator (in a known plaintext attack), or
of some bits of the ciphertext sequence (in a ciphertext only attack).

\subsection{Time Domain Analysis of Filter Generators}
Let $s_t$ be an $m$-sequence with maximum period $2^{m}-1$ generated from an  LFSR whose length is $m$, then $z_t$ is the output sequence of a filter generator

\begin{equation}\label{eq:FG}
z_t =   f (s_0,s_1,s_2,.....,s_{t-1}) \;\;\;\forall t\geq 0;
\end{equation} 
where $s_0$ are inputs of nonlinear function $f$ coressponding to taps of an LFSR as shown in Figure below.

\begin{figure}[h!]\label{fig:FG1}
	  \centering
		      \includegraphics[scale=0.60]{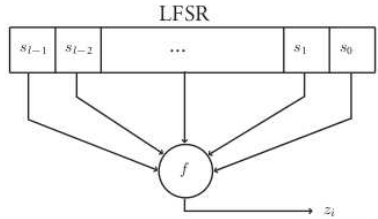}
	  \caption{A Simple Filter Generator}
	\end{figure}
In order to obtain a keystream sequence having good cryptographic properties, the filtering function $f$ should be balanced (i.e., its output should be uniformly distributed), should have large algebraic degree with large correlation and algebraic immunity.

Now we present few important facts related to design parameters of filter generators in time domain~\cite{canteaut2011stream}: 

\begin{enumerate}

\item The period of output stream $z_t$ is $2^{m}-1$ where $m$ is degree of feedback polynomial of LFSR which is primitive in most cases. The requirement of long period is a direct consequence of shift property of $m$-sequences as defined in (\ref{eq:shift}) within a linear space $G(f)$. 
\item The output sequence of a filter generator $z_t$ is a linear recurring sequence whose linear complexity $L(z)$ is related to length of LFSR and algebraic degree of nonlinear filtering function $f$. For any LFSR defined over $GF(2)$ with primitive feedback polynomial, an upper bound of linear span is
\begin{equation}\label{eq:LC}
L(z) \leq   \sum^{d}_{i=1} (^{l}_{i} )
\end{equation}
where $d$ is the degree of filtering function $f$ and $l$ is length of LFSR. Thus for a higher $L$ in sequence generator designs, $d$ and $l$ are chosen to be as high as possible.  

\item LFSR feedback polynomial should not be sparse.
\item The tapping sequence should be such that the memory size (corresponding to the largest gap between the two taps) is large and preferably close to its maximum possible value. This resists the generalized inversion attack ~\cite{canteaut2011stream} which exploits the memory size.  
\item If the LFSR tabs are equally spaced, then the lower bounds of linear complexity will be 
\begin{equation}\label{eq:LS-equaltaps}
L(z) \geq    ( ^{l}_{d}  )
\end{equation}

\item The boolean function used as a filtering function must satisfy following criteria to be called as good cryptographic function. For details, readers may refer to  ~\cite{cusick2009cryptographic}.
\begin{enumerate}
\item The boolean function should have high algebraic degree. A high algebraic degree increases the linear complexity of the generated sequence.

\item It should have high correlation immunity which is defined as measure of degree to which its output bits are correlated to subset of its input bits. High correlation immunity forces the attacker to consider several input variables jointly and thus decreases the vulnerability of divide-and-conquer attacks.

\item It should have high non-linearity which is related to its minimum distance from all affine functions.  A high nonlinearity gives a weaker correlation between the input and output variables. This criteria is in relation to linear cryptanalysis, best affine appproximation attacks, and low order approximation attacks. Moreover, high non-linearity resists the correlation and fast correlation attacks ~\cite{meier1989fast} by making the involved computations infeasible.
\item The recent developements in cryptanalysis attacks introduced the most significant cryptographic criteria for boolean functions termed as algebraic immunity. There should not be a low degree function $g$ for $f$ which  satisfies $f * g = 0 $ or $(f+1) * g = 0 $. Algebraic immunity is defined as minimum degree function $g$ for any bollean function $f$ satisfying the said criteria. Attacks exploiting this weakness in boolean functions are called algebraic ~\cite{courtois2003algebraic} and fast algebraic attacks  ~\cite{courtois2003fast},~\cite{canteaut2006open}. 
\item A more recent attack on filter generators     ~\cite{ronjom2007new}  exploiting the underlying algebraic theory suggests an updated estimate of degree of non linear boolean functions to resist the new kind of attack on filter generators.  
\end{enumerate}

\end{enumerate}
\subsection{Frequency Domain Analysis of Filter Generators}

This section provides DFT based analysis of filter generators which has been known in public domain for few years. Few additonal results of our frequency domian analysis has also been included here. 
\begin{enumerate}
\item Irrespective of time domain requirement of a long period for a filter generator, DFT simplifies the associated high computational problem to any $k$-th component of Fourier spectra through the relation
\begin{equation}\label{DFT shift eq-2}
  \alpha^{\tau} =  ( Z_{k} . A_{k}^{-1} )^{k^{-1}}
\end{equation}
 where $\tau$ determines the exact amount of shift between $a_t$ and $z_t$ and k is index of any one component of DFT spectra. 
\item To obtain specific component of Fourier spectra and limit the computational complexity within affordable bounds, Fast Discrete Fourier spectra attacks ~\cite{gong2011fast} propose an idea of selecting a suitable filter polynomial $q(x)$, an LTI system, to pass only specific spectral points and restrict all others which are nulled to zero. To illustrate, we mention an important lemma here without proof.
\begin{lemma}
\label{LTI-lemma}
Let $q(x)$ be a polynomial defined over $GF(2)$ with period $r$ as $q(x) = \sum_{i=0}^{r} c_{i}x^{i}$. We apply $s_t$ to LTI system having a function $q(x)$ and $z_t$ is the output sequence. By theory of LTI system resposne to any arbitrary signal and convolution in time domain, we have \\
\begin{equation}\label{LTI-eq-1}
  z_t =  \sum_{i=0}^{r} c_{i}s_{i+t}\;\;\;t = 0,1,...,n-1
\end{equation}
Converting the relation into frequnecy domain, we get
\begin{equation}\label{LTI-eq-2}
  Z_k =  q(\alpha^k)S_k\;\;\;k = 0,1,...,n-1
\end{equation}
where $q(\alpha^k)$ is infact $Q_k$; a way of interpreting DFT in terms of polynomials.
\end{lemma}

\item Selection of polynomial $q(x)$ as a filter function is  discussed at length in~\cite{gongcloser}. To summarize, steps of computations are :-
\begin{enumerate}
\item Computing minimum polynomial $g(x)$ of output stream $z_t$ by generating a refernce sequence $a_t$  which infact is a shifted version of $z_t$.
\item Selecting  $k$ amongst the coset leaders such that
gcd $(k,n)=1$ and $g(\alpha^k)=0$.
\item Computing $k$-decimated sequence of reference stream $a_t$ as $c_t = a_{kt}$.
\item Using Berlekamp-Massey algorithm, computing minimum polynomial $g_{k}(x)$.
\item Computing $q(x)$ through a relation
\begin{equation}\label{eq:filter poly}
  q(x) = \frac{g(x)}{g_{k}(x)}
\end{equation}
\end{enumerate}   

\item Selection of $q(x)$ can be done by direct factorization of minimum polynomail $g(x)$. 
\item The DFT over binary fields can be computed for a filter generator without requiring entire sequence. Authors in~\cite{gong2011fast} have provided detailed algorithm for computing DFT for sequences with fewer bits (equal to linear span of that sequence or even lesser) as compared to the total period of the sequence. we also propose an alternate approach to select $q(x)$ in the next section dealing with the combiner generators.  
\item The consequence of Fast Discrete Fourier Spectra attacks on sequence generators introduced a new design criterion of spectral immunity. This criteria implies that in order to  resist the selective DFT attack, the minimal polynomial of an output sequence of an LFSR based key stream generator should be irreducible.  
\item Recent studies have shown that low spectral weight annihilators are essential for Fast Discrete
Fourier Spectra attacks~\cite{wang2012annihilators}.

\end{enumerate}

\section{Transformed Domain Analysis of Combiner Generators}

A combiner generator consists of number of LFSRs which are combined by a nonlinear Boolean function.  The Boolean function $f$ is called the combining function and its output is the keystream. The Boolean function $f$ must have high algebraic degree, high nonlinearity and preferably a high order of correlation immunity. 

\subsection{Time Domain Analysis of Combiner Generators}
Consider a combiner generator consisting of $l$  LFSRs as shown in Figure ~\ref{fig:CG1}. 

\begin{figure}[h!]\label{fig:CG1}
	  \centering
		      \includegraphics[scale=0.25]{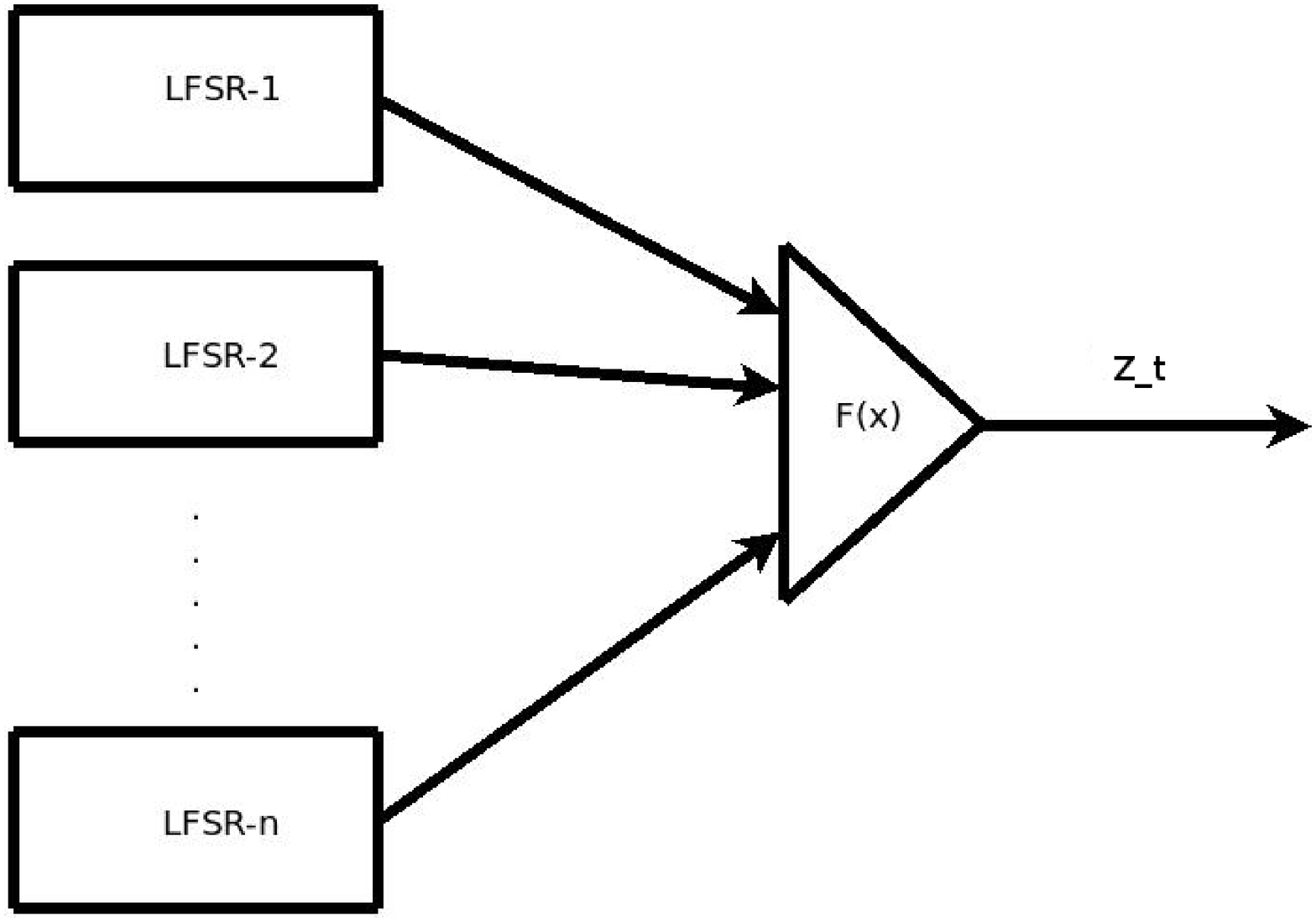}
	  \caption{A Simple Combiner Generator}
	\end{figure}
We denote the output sequence of the $i$-th LFSR as $a^i$, its minimal polynomial $g_i (x)$ and its length $r_i$ where $i = 1, 2, . . . , l$. We will assume here that all $r_i$'s are mutually coprime which is generally the case for the combiner generators. Let $f(a_t^1,a_t^2,....,a_t^l)$ be a non linear function where $f: GF(2^l) \rightarrow GF(2)$ takes input from $l$ LFSRs and  produces the key stream $z_t$ as

\begin{equation}
	  z_t = f(a_t^1,a_t^2,....,a_t^l)
	  \end{equation}

The criteria for selecting LFSRs and boolean function in a combiner generator are mostly similar to filter generators. However, few additions/differences related to design parameters are as follows:
\begin{enumerate}
\item Contrary to filter generators where a single LFSR with larger length is used, combiner generator employs multiple  LFSRs with comparatively smaller lengths. The requirement of primitive polynomials as feedback polynomials stays as such for each LFSR.
\item  The period of output keystream becomes $lcm (r_1,r_2,...,r_l)$.
\item Since the combining function involves both Xoring and multiplication operation, few properties of linear complexities are reproduced here from~\cite{brynielsson1986linear},~\cite{rueppel1987products}  and~\cite{herlestam1986functions}.  Considering linear complexities of LFSR sequences $a_t^1,a_t^2,....,a_t^l$ as $L_1,L_2,....,$ and $L_l$:
\begin{enumerate}
\item the linear complexity of the sequence $s_t = (a_{t}^1 + a_t^2 )$ satisfies

\begin{equation}\label{eq:LC-sum}
L (s_t) \leq L(a_{t}^1) + L(a_t^2),
\end{equation}
 the equality holds if and only if the minimal polynomials of $a_{t}^1$ and $a_{t}^2$ are relatively prime.
 
\item From~(\ref{eq:LC-product}), the linear complexity of the sequence $s_t = (a_{t}^1 . a_t^2 )$ satisfies

\begin{equation}\label{eq:LC-product-2}
L (s_t) \leq L(a_{t}^1) . L(a_t^2)
\end{equation}
  
\item The linear complexity of $z_t = f(a_t^1,a_t^2,....,a_t^l)$ satisfies
\begin{equation}\label{eq:LC-combiner}
L (z_t) = f(L_1,L_2,.....,L_l)
\end{equation}

\end{enumerate} 
\item The novel observations regarding fixed patterns of LFSRs,   cyclic structures existing in finite fields and their interpretation through CRT imply that index of observed keystream bits in any reference stream directly reveals the initial state of all LFSRs. CRT based interpretation of LFSR sequences in relation to their periods thus reiterates the requirement of long period for sequences of combiner generators.   
\end{enumerate} 

\subsection{Frequency Domain Analysis of Combiner Generators}
In this section, frequency domain analysis of combiner generators is presented. the application of selective DFT attacks on combiner generators is discussed in detail with some novel observations. Developing on the theory of selective DFT attack, a new efficient methodology is proposed to identify the initial states of all LFSRs of combiner generators.

In section-4.2, the direct mapping of product sequence between time and frequency domains was demonstrated. Likewise, the Xoring operation existing in most of the boolean functions also demonstrates similar mapping. For application of boolean functions in combiner generators, following is important:

\begin{enumerate}
\item For the product terms of a non-linear boolean function, a term of DFT spectra of the product stream is nonzero if and only if all the component DFT terms are nonzero. With known non-zero DFT points for $a_t^1,a_t^2,....,a_t^l$, CRT can be used to determine non zero  points of DFT spectra of $s_t$ directly as:
   
   \begin{eqnarray*}
      x  &\equiv& k_1\mbox{\; (mod} \mbox{\;}r_1) \\
      x  &\equiv& k_2\mbox{\; (mod} \mbox{\;}r_2)\\
      .  &\equiv& ...\mbox{\; ....} \mbox{\;}....\\
      x  &\equiv& k_l\mbox{\; (mod} \mbox{\;}r_l)
               \end{eqnarray*}
   where $k_1, k_2.....k_l$ denote non zero index positions of $a_t^1,a_t^2,....,a_t^l$ and $x$ is the position of non-zero componenet of DFT spectra of $s_t$ within its period $r$.

\item For terms of a non-linear boolean function being Xored, any DFT spectra term will be nonzero for which number of nonzero component DFT terms are odd. 
 
\end{enumerate}
Let us explain these facts through an example.
\begin{example}\label{exmp-2}
Consider a combiner generator consisting of 3 LFSRs with primitive polynomials as $g_{1}(x) = x^{2}+x+1,\; g_{2}(x) = x^{3}+x+1  $ and $g_{3}(x) = x^{5}+x^{2}+1$. The outputs of LFSRs in this case are m-sequences, denoted as $a^{1}_{t},a^{2}_{t}$ and $a^{3}_{t}$ respectively. Output stream of the generator is denoted as $z_{t}$ and a nonlinear function $f(x)$ is 
  	\begin{equation*}
  	 f(x) = a^{1}_{t}.a^{2}_{t} + a^{2}_{t}.a^{3}_{t} + a^{1}_{t}.a^{3}_{t} 
  	 \end{equation*}
  	where period of $z_{t}$ in this case becomes 651 as $\mbox{lcm}(3,7,31)= 651$.
 
Mapping of operations is demonstrated in spectral domain using~\ref{DFT eq} as:
 \begin{enumerate}
 \item DFT of individual LFSRs:
 \begin{enumerate}
 \item DFT of $a^{1}_{t}$:$ 0,1,1$
 \item DFT of $a^{2}_{t}$: $0,0,0,\alpha^{4},0,\alpha^{2},\alpha$
 \item DFT of $a^{3}_{t}$: only five non-zero DFT terms are at indices $15, 23, 27, 29$ and $30$ with values $\alpha^{29}, \alpha^{30}, \alpha^{15}, \alpha^{23}$ and $\alpha^{27}$ respectively.
  \end{enumerate}
  \item DFT of product streams:
  \begin{enumerate}
 \item DFT of $a^{1}_{t}a^{2}_{t}$: Corresponding to minimum polynomial of $x^{6}+ x^{4}+x^{2}+x+1$ and period = $21$, six non-zero components are $\alpha^{9}, \alpha^{18}, \alpha^{15}, \alpha^{18}, \alpha^{9}$  and $\alpha^{15}$ at indices $5, 10, 13, 17, 19$ and $20$. These indices can be easily determined while working in component fields of $GF(2^2)$ and $GF(2^3)$ by using CRT calculations as discussed in Section 4.2.
 \item DFT of $a^{2}_{t}a^{3}_{t}$: With minimum polynomial $x^{10} + x^5 + x^4 + x^2 + 1$ and period = $93$, ten non-zero components are at indices $23, 29, 46, 58, 61,\\ 77, 85, 89, 91$ and $92$.
 
 \item DFT of $a^{3}_{t}a^{1}_{t}$: Similarly with minimum polynomial $x^{15} + x^{12} + x^{10} + x^7 + x^6 + x^2 + 1$ and period = $217$, fifteen non-zero components are at indices $27, 54, 61, 89, 108, 122, 139, 153, 178, 185, 201, 209, 213, 215$ and $216$.
 \item To verify the established facts, DFT of product of all three streams has also been analyzed. For $a^{1}_{t}a^{2}_{t}a^{3}_{t}$ having a minimum polynomial of $x^{30} + x^{25} + x^{24} + x^{20} + x^{19} + x^{17} + x^{16} + x^{13} + x^{10} + x^9 + x^8 + x^7 + x^4 + x^2 + 1$ with period = $651$, thirty non-zero DFT components are at indices $61, 89, 122, 139, 178, 185, 209, 215, 244, 271, 278, 325, 356, 370, 395, 418, 430,\\ 433, 461, 488, 523, 542, 556, 587, 619, 635, 643, 647, 649$ and $650$. All these indices can be determined directly from knowing the individual DFTs of three LFSRs separately. For instance,
 \begin{eqnarray*}
      x  &\equiv& 1\mbox{\; (mod} \mbox{\;}3) \\
      x  &\equiv& 3\mbox{\; (mod} \mbox{\;}7)\\
      x  &\equiv& 15\mbox{\; (mod} \mbox{\;}31)
               \end{eqnarray*}
               gives result of 325 which exists amongst thirty non-zero DFT computations as well.
  \end{enumerate}
  \item To see impact of Xor operation in frequency domain, DFT of $z_t$ with minimum polynomial $x^{31} + x^{29} + x^{28} + x^{27} + x^{24} + x^{23} + x^{22} + x^{20} + x^{18} + x^{17} + x^{16} + x^{15} + x^{13} + x^{11} + x^{10} + x^9 + x^8 + x^7 + x^5 + x^4 + x^2 + x^1 + 1$ is computed. Results reveal that all indices where number of non-zero DFT terms for three product streams i.e. $a^{1}_{t}a^{2}_{t}$, $a^{2}_{t}a^{3}_{t}$ and $a^{1}_{t}a^{3}_{t}$ are odd, resulting DFT term is non-zero. For instance indices at 5, 10, 13, 17, 19 and 20 where only DFT of $a^{1}_{t}a^{2}_{t}$ term is non-zero, resulting DFT for $z_t$ is also non-zero. Similar is the case for other indices.
 \end{enumerate}
\end{example}

\subsubsection{Selective DFT Attacks on Combiner Generators.}
In this subsection, possibilty of extending selective DFT attack on combiner generators has been discussed. The attack algorithm on non-linear filter generator has been explained in~\cite{gong2011fast} and~\cite{gongcloser}. However, direct application of selective DFT attack on combiner generators has few limitations with regard to underlying design of these type of sequence generators. For simplicity, case of $m = L(a_t)$ has been considered here where $m$ is the number of known bits of key stream and $a_t$ is the coordinated scaled sequence of key stream. 
\begin{enumerate}
\item As combiner generators entail multiple LFSRs, determination of element $\beta \in GF(2^m)$ through coordinated scaled sequence doesnot lead to initial states of all LFSRs directly. 
\item In precomputation stage of the algorithm as discussed at length in~\cite{gongcloser}, $k$-decimation sequence of LFSR output sequence is computed followed by applying Berlekamp-Massey algorithm on it to determine its associated minium polynomial $g_k(x)$. With the help of this sequence $c_t$, $m$ x $m$ circulant matrix is obtained as follows   

\[ M = \left[ \begin{array}{cccccc}
c_0 & c_1 & c_2 & ... & c_
{m-2} & c_{m-1} \\
c_1 & c_2 & c_3 & ... & c_
{m-1} & c_{m} \\
c_2 & c_3 & c_4 & ... & c_
{m} & c_{m+1}\\
. & . & . & ... & . & . \\
. & . & . & ... & . & . \\
. & . & . & ... & . & . \\
c_{m-1} & c_{m} & c_{m+1} & ... & c_
{2m-3} & c_{2m-2} \end{array} \right]\]

\item In~\cite{gong2011fast}, this matrix is termed as coefficient matrix $H$ and $g_k(x)$ is minimum polynomial for $\alpha^k$ where $a_t$ is determined from~(\ref{eq:trace}) or through frequency component as:

\begin{equation}
  	 a_t = \sum_{k} Tr_{1}^{n_k}(A_{k}\alpha^{tk})\;\;t = 0...n-1
  	 \end{equation}
\item We propose another approach to compute $g_k(x)$ through factoring $g(x)$. In this case, step involving Berlekamp-Massey algorithm on $k$- decimated sequence will no longer be required. Setting up matrix $M$ or $H$ in this case is by direct initializing of corresponding LFSR from any random state.
\item The output of selective DFT algorithm produces $\beta = \alpha^{\tau}$. Using this value of $\tau$, left shift value for each LFSR sequence is determined by applying modular computations of CRT with respect to individual periods $r_i$'s of LFSRs  as:
	   \begin{equation*}
	   \tau_1\;  \;\equiv\; \; \tau \mbox{\;(mod\;}r_1) \end{equation*}
	   \begin{equation*}
	   \tau_2\;  \;\equiv\;  \tau \mbox{\;(mod\;}r_2)\end{equation*}
	   \begin{equation*}...\;\;\;\;...\;\;\;\;...\;....\end{equation*}
	   \begin{equation*}
	  \tau_l\;  \;\equiv\; \; \tau \mbox{\;(mod\;}r_l)\end{equation*}
	   
		\item Determine the initial state of each LFSR by applying individual shift values $\tau_i$'s to LFSRs sequences by 
		using (\ref{eq:trace}) within each field $GF(2^m)$ where 
\begin{eqnarray*}
      \beta_i  &=& \alpha^{\tau_i} \\
      b_{t}^{i} &= &Tr_{1}^{n} (\beta_i \alpha_{i}^{t}) 
        	\end{eqnarray*} 
\item If $\tau$ is directly applied to each LFSR, number of computations involved in shifting LFSR sequence is of the order to  $\mathcal{O}(\tau)$. CRT based interpretation of LFSR shifts in initial states with respect to their periods save the last step computations of selective DFT attacks.   
\end{enumerate}  
Let us demonstrate these observations through an example.
\begin{example}
Consider the same combiner generator as in Example~\ref{exmp-2}. Suppose we have only $31$ bits of  keystream $z_t = [1011110001111010111001011010111]$. With a known structure of the generator, our attack will determine the initial state of three LFSRs as follows:
\begin{enumerate}
\item Initially, possibility of success of selective DFT attack on given combiner generator will be determined. Through Berlekamp-Massey algorithm, minimum polynomial $g(x)$ of keystream $z_t$ will be computed. Applying factorization algorithm on $g(x)$ gives its three factors as $g_{1}(x) = x^{6}+ x^{4}+x^{2}+x+1$, $g_{2}(x) = x^{10} + x^5 + x^4 + x^2 + 1$ and $g_{3}(x) = x^{15} + x^{12} + x^{10} + x^7 + x^6 + x^2 + 1$. So the selective DFT attack on this combiner generator is possible.
\item  Generating a reference sequence $a_t$ and decimating it with $k = 58$ with gcd$(58,651)=1$ produces a coefficient sequence $c_t = 011010101110$.   
\item Applying Berlekamp-Massey algorithm on $c_t$ gives its associated minumum polynomial of $g_k(x) = x^6 + x^4 + x^2 + x + 1$.
\item Circulant matrix will thus be of dimension $6$ x $6$ as:
\[ M = \left[ \begin{array}{cccccc}
0 & 1 & 1 & 0 & 1 & 0 \\
1 & 1 & 0 & 1 & 0 & 1 \\
1 & 0 & 1 & 0 & 1 & 0 \\_
0 & 1 & 0 & 1 & 0 & 1 \\
1 & 0 & 1 & 0 & 1 & 1 \\
0 & 1 & 0 & 1 & 1 & 1 \\
1 & 0 & 1 & 1 & 1 & 0 \\
 \end{array} \right]\]

\item From~(\ref{eq:filter poly}), filter polynomial $q(x) = x^{25} + x^{22} + x^{19} + x^{17} + x^{10} + x^{9} + x^{8} + x^{5} + 1$.
 
\item The same results could be achieved from our observations as follows:
\begin{enumerate}
\item From DFTs of $a^{1}_{t}$, $a^{2}_{t}$ and $a^{3}_{t}$  while working in finite fields of $GF(2^3)$, $GF(2^3)$ and $GF(2^5)$, $k=58$ can be directly calculated to be non-zero index for DFT of $z_t$. 
\item Factorization done initially to check applicability of our selective DFT attacks already showed $g_1(x) = g_k(x)$ as one of the factor.
\item We thus obtain the same $q(x)$ as in step-5 above.
\end{enumerate}

\item Computations of selcetive DFT algorithm give the result of  $\tau = 19$.
	  
\item Finally, we can determine the left shift value in sequences for each LFSR by applying modular computations of CRT with respect to individual periods of LFSRs $r_i$'s as
	   \begin{equation*}
	   \tau_1\;  \;\equiv\; \; 19 \mbox{\;(mod\;}3) \end{equation*}
	   \begin{equation*}
	   \tau_2\;  \;\equiv\;  19 \mbox{\;(mod\;}7)\end{equation*}
	   \begin{equation*}
	  \tau_l\;  \;\equiv\; \; 19 \mbox{\;(mod\;}31)\end{equation*}
	   
\item Having determined the exact shift value for each LFSR, their initial states will be computed by 
		using (\ref{eq:trace}) within each field $GF(2^m)$ where 
		\begin{eqnarray*}
      \beta_i  &=& \alpha^{\tau_i} \\
      b_{t}^{i} &= &Tr_{1}^{n} (\beta_i \alpha_{i}^{t}) 
        	\end{eqnarray*}
		
	\item The initial fills of LFSRs with 1 left shifts in $a_t^1$, 5 left shifts in $a_t^2$ and 19 left shifts in $a_t^3$ gives:
 	
\begin{table}[H]
\begin{center}
\caption[Sample Table]{Initial States of 3 LFSRs}
\begin{tabular}{|c| c| }
\hline
& Initial State 
   \\ \hline
 LFSR-1 & 10   \\ \hline
 LFSR-2 & 101   \\ \hline
 LFSR-3 & 01111   \\ \hline

\end{tabular}

\label{tab:initial states-3LFSRs}
\end{center}
\end{table}	 
\end{enumerate}

\end{example}
We are thus able to recover exactly the initials states of all the LFSRs through our novel approach of interpretting fixed patterns in LFSRs sequences through CRT. Although the proposed approach was demonstrated on an illustrative example, but it holds true for any configuration of combiner generators. The same was tested with different non-linear combiner functions as well as with different number of LFSRs.    
\subsection{Complexity Comparisons}
Let us see the advantage of frequency domain analysis of LFSR based combiner generator over its time domain analysis. We will discuss computational complexity in relation to most common attacks also.
For a combiner sequence having period $N$ with $l$ number of constituent LFSRs, complexity of Exhaustive search is $2^{m_{1} + m_{2} + ...+ m_{l} -1}$ where as for correlation attack complexity reduces to $2^{m_{1}-1}+ 2^{m_{2}-1}+ ...+ 2^{m_{l}-1}$. To compute the complexity in case of selcetive DFT attacks on combiner generators, calculations for preprocessing and actual attack stage are~\cite{gong2011fast}:
\begin{enumerate}
\item Preprocessing Stage. The computations during this stage are sum of following:
\begin{enumerate}
 
\item The complexity of computing minimum polynomial $g(x)$ by applying Berlekamp-Massey algorithm on $s_t$ is $\mathcal{O}(L\; log_2(L))$ Xor operations where $L$ is linear complexity of $s_t$. Incase of using a relation $g(x) = \prod_{k\in I} g_k(x)$, complexity is $|I|[log_{2}(L)]^3$.
\item The complexity of computing $g_k(x)$ is $m\;log_2(m)^2$ for each $k$. For $k\in {I}$ with $I$ representing set of coset leaders, complexity will be $|I|[m\;(log_2\;m)^2]$.
\item The complexity of computing  $g(\alpha^k)$ is $\mathcal{O}(N\;\eta(m)+m\;log_2(N))$ $GF(2)$ operations where $\eta(m) = (m\;log\; m \;log\; log\; m)$ and $N$ is the degree of polynomial $g(x)$.
\end{enumerate}
\item Attack Stage. The complexity of this stage is sum of following:
\begin{enumerate}
\item Passing our sequence $s_t$ from LTI filter $q(x)$ is actually time convolution of $q(x)$ and $s_t$ which costs $\mathcal{O}(L)$ GF(2) operations.
\item Last step of DFT spectra attack giving the output $\beta$ is solving the system of linear equations over GF(2) in $d$ unknowns which has the complexity of utmost $\mathcal{O}(d^w)$, where $w$ is Strassen's reduction exponent $w = log_2(7) \approx 2.807$.
\item Determining left shift values for each LFSR through $\tau_i \equiv \tau\;\mbox{mod}\; r_i$ and computing initial states of LFSRs are of order $\mathcal{O}(l)$ each and are negligible, where $l$ is number of LFSRs. 
  
\end{enumerate}
\end{enumerate}
Having established the comparisons of complexities between exhasutive serach, correlation attack and selective DFT attack, let us map these to our example-3 above. Exahustive search costs utmost $2^{2+3+5-1}\; = \;2^9 \;\approx \mathcal{O}(512)\;GF(2)$ computations. With Probabilities of $P(z = a) = 6/8$, $P(z = b) = 6/8$ and $P(z = c) = 6/8$, correlation attack largerly reduces the complexity to $2^{2-1} + 2^{3-1} + 2^{5-1} \; \approx \; \mathcal{O}(21)$ $GF(2)$ operations. Incase of selective DFT attacks, complexity of preprocessing stage is  $\mathcal{O}(279)$ and of attack stage is $\mathcal{O}(150)$ making total of $\mathcal{O}(429)$ $GF(2)$ operations. Thus it can be clearly stated that after one time preprocessing computations of selective DFT attacks, complexity of actual attack stage with $\mathcal{O}(150)$ is promisingly less as compared to exhastive search attack. However, correlation attacks and their faster variants are more efficient than DFT attacks in special scenerios where underlying combining function is not correlation immune. Incase of correlation immune combining functions, selective DFT attacks still provide propitious results and are advantageous over the correlation attacks. 

\section{Applicability of Fast Discrete Fourier Spectra Attacks on A5/1 Algorithm}
In this section applicability of fast discrete fourier spectra attacks on A5/1 algorithm is discussed. For clarity of context description of algorithm structure is given first followed by discussion on possibility of selective DFT  attacks on it.
\subsection{Description of A5/1}
A5/1 is a stream cipher built on a clock controlled combiner generator. Being a Global System for Mobile Communications (GSM)  encryption algorithm, it has been intensively  analyzed and is considered weak becuase of number of succesful attacks now. Details can be found in~\cite{barkan2006conditional}, ~\cite{ekdahl2003another}, and~\cite{gendrullis2008real}. The keystream generator consists of three LFSRs, $R1,R2$ and $R3$ of lengths 19, 22 and 23 respectively as shown in figure-3. 
\begin{figure}[h!]\label{fig:a5}
	  \centering
		      \includegraphics[scale=0.87]{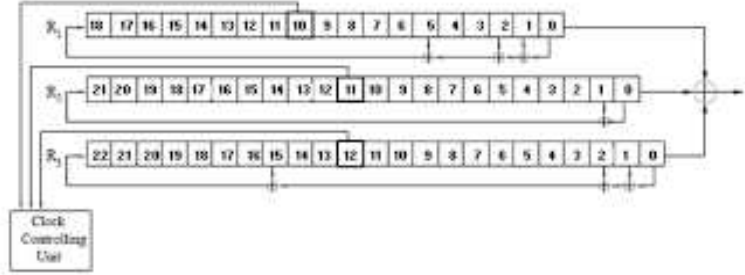}
	  \caption{Design of A5/1 Algorithm}
	\end{figure}


The taps of three LFSRs correspond to primitive polynomials $x^{19}+x^{5}+x^{2}+x+1$, $x^{22}+x^{1}+1$ and $x^{23}+x^{15}+x^{2}+x+1$ and therefore, LFSRs produce maximum periods. The registers are clocked irregularly based on decision of a majority function having input of  clocking bits 8, 10 and 10 of registers $R1,R2$ and $R3$ respectively. It is a type of stop and go clocking where those LFSRs are clocked whose most significant bit (msb) matches to output bit of the majority function. At each clock cycle either two or three registers are clocked, and that each register is moved with probability $3/4$ and stops with probability $1/4$. The running key of $A5/1$ is obtained by XORing of the output of the three LFSRs. The process of keystream generation is as follows:-
\begin{enumerate}
\item Starting with key initialization phase, all LFSRs are initialized to zero first. 64-bit session key $K = (k_0 , . . . , k_{63})$ and publically known 22-bit frame number serve as initialization vector. In this phase all three registers are clocked regularly for $86$ cycles during which the key bits followed by frame bits are xored with the least significant bits of all three registers consecutively.

 \item  During second phase, the three registers are clocked for $100$ additional cycles
with the irregular clocking, but the output is discarded.
\item Finally, the generator is clocked for $228$  clock cycles with the irregular clocking producing the $228$ bits that form the keystream. 114 of them are used to encrypt uplink traffic from $A$ to $B$, while the remaining $114$ bits are used to decrypt downlink traffic from $B$ to $A$. A GSM conversation is sent as a sequence of frames, where one frame is sent every $4.6$ ms and contains $114$ bits. Each frame conversation is encrypted by a new session key $K$.

\end{enumerate}
\subsection{DFT Attacks on A5/1}
Applicablity of selective DFT attacks is preconditioned with following two separate cases:-
\begin{enumerate}
\item Case 1. If the minimal polynomial of output keystream $z_t$ is reducible, algorithm 1 described in~\cite{gong2011fast} is directly applicable.
\item Case 2. If another sequence $y_t \in GF(2^n)$ is determined such that $v_t = z_t * y_t$, where $*$ is a term wise product, and $L(v_t) +L( y_t) \leq L(z_t)$ and $L(v_t + y_t)+L(y_t)$ $\leq L(z_t)$, algorithm 2 described in~\cite{gong2011fast} is applicable on $z_t$. 

\end{enumerate}
Our analysis reveals that selective DFT attacks on A5/1 algorithm are not applicable. Detailed results are being published at other forum shortly 
    
 \subsection{DFT Attacks on E0 Cipher}
Selective DFT attacks on E0 cipher are possible with modifications in equations derived in~\cite{armknecht2002linearization}. Our results on E0 cipher are being published at some other forums shortly.
\section{Conclusion}
In this report, we presented a transform domain analysis of LFSR based sequence generators. The inherent peculiarities of the LFSR product sequences were evoked through novel patterns identified with the help of a CRT based approach. These findings were then extended to the filter generators and more particularly to the combiner generators. An effort was made to establish the mapping of different operations from time domain to frequency domain. Novel results on fixed shift patterns of LFSRs, their relationship to cyclic structures in  finite fields and CRT based interpretation of these patterns have been exploited to reduce the computations required in the last stage of DFT spectral attacks attacks on combiner generators. Subsequent to the transform domain analysis of basic components of stream ciphers and discussion on applicability of fast discrete fourier attacks on A5/1 algorithm, DFT based analysis of combiners resistant to correlation attacks are considered as interesting cases for their analysis in frequency domain and some initial results have shown good promise in this regard.  


\bibliographystyle{plain}
\bibliography{asad}

\end{document}